\documentclass[article]{llncs}
\usepackage[T1]{fontenc}
\usepackage{graphicx}
\usepackage{hyperref}
\usepackage{color}

\usepackage{amsmath, amssymb}
\usepackage{algorithm}
\usepackage{algpseudocode}
\usepackage{tabularx}
\usepackage{threeparttable}
\usepackage{orcidlink}
\begin{document}
\title{Deanonymizable Scoped Linkable Ring Signatures}
%
%
\author{Montassar Naghmouchi \orcidlink{0000-0003-3467-7514} and Maryline Laurent \orcidlink{0000-0002-7256-3721}}

\authorrunning{M. Naghmouchi and M. Laurent}
%
\institute{ SAMOVAR, Télécom SudParis, Institut Polytechnique de Paris}
\email{}
\maketitle              
\begin{abstract}
Ring signatures offer highly desirable privacy features like anonymity and ad-hoc group formation with high autonomy, but partially lack linkability and accountability features required for strict use-cases like consent management in healthcare. Existing signature schemes fail to natively integrate scoped linkability with decentralized accountability defined here on-demand unconditional deanonymization of signers without reliance on a trusted party (opener) or a separate cryptographic commitment. We introduce Deanonymizable Scoped Linkable Ring Signatures (DSLRS) to address this gap. DSLRS builds on previous linkable and accountable signature schemes and their cryptographic primitives to construct a new signature scheme that natively implement both privacy and trust features. Scoped linkability in DSLRS is achieved using "scopes" (context identifiers) and dynamic key images, effectively providing ring-independent scoped linkability and unlinkability across different scopes. DSLRS relies on k-of-N decentralized deanonymization network nodes to extract the signer's public key (identity) from ElGamal components built in the signature, achieving ring-independent decentralized accountability. The full DSLRS algorithms, formal property definitions and proofs, signature evaluation and experimental results from an implementation, and a blockchain-based instantiation and application of our scheme in consent management for clinical trials are provided. DSLRS is proved under the ECDLP and DDH hardness assumptions in the Random Oracle Model (ROM).

\vspace{1em}
\noindent \textbf{Keywords:} Ring Signatures, Linkability, Accountability, Decentralization, Consent Management, Blockchain
\end{abstract}
\section{Introduction}
Ring signatures enable ad-hoc group (ring) formation and allow signers to remain anonymous within a given ring. Intended to be used by whistle-blowers \cite{rivest}, the first ring signatures did not account for the need to link two signatures together, or for the need of accountability, where the identity of the signer needs to be revealed. The absence of a group manager made it even harder to implement such features. Linkable ring signatures provide linkability that extended the use-cases of ring signatures into domains that require detecting double signing, such as e-voting \cite{LRS}, preventing double-spending in blockchains \cite{monero}, etc. Later on, accountable variants \cite{Xu-acc} provided revocation of anonymity either on-demand \cite{Boot} or upon signer misbehavior \cite{TRS}. This presented further application possibilities for ring signatures, like online anonymous forums and anonymous auctions, or simply introducing accountability to previous use-cases - like holding dishonest whistle-blowers accountable.

While previous applications focus on isolated transactions and communications, healthcare applications require strict data governance. Consequently, consent management in healthcare and clinical trials requires high privacy standards combined with accountability. For medical and legal purposes, an anonymous consent giver can be traced back and deanonymized on demand if needed and under specific conditions like the need to re-identify the patient. For such use-case, a signature scheme is required to provide \textbf{signer indistinguishability} to provide anonymity for signers, \textbf{scoped linkability} in order to allow patients to be traced within the same scope (context, like a given research project), \textbf{cross-scope unlinkability} so they remain unlinkable in different contexts (for example across many research projects) and \textbf{accountability} which is \textbf{on-demand deanonymization} to satisfy re-identification needs. Note that the need for \textbf{cross-scope unlinkability} is highlighted in cases of \textbf{on-demand deanonymization} since it allows the reveal of the signer's identity only in certain contexts while preserving their global anonymity.
Furthermore, to avoid reliance on centralized authorities, \textbf{decentralized deanonymization} is desired.

Existing schemes fail to satisfy all criteria simultaneously. We therefore introduce Deanonymizable Scoped Linkable Ring Signatures (DSLRS).

\underline{Our contributions} are: (1) A new ring signature construction and algorithms that satisfies indistinguishability, scoped linkability, cross-scope unlinkability and embedded decentralized accountability; (2) Formal security definitions and proofs under ECDLP and DDH assumptions in the Random Oracle Model (ROM); (3) A consent management use-case and an effective blockchain instantiation of the proposed scheme.

The rest of the paper is structured as follows: Section \ref{background} provides a related work comparison, Section \ref{sec3} provides an overview and threat model in \ref{overview}, preliminaries and security assumptions in \ref{preliminaries}, scheme definition of DSLRS and the full DSLRS algorithms in Section \ref{scheme-def}. We also provide an evaluation and experimental results of the algorithms in \ref{eval}. In Section \ref{analysis} we provide formal definitions of our signature scheme properties using cryptographic game sets. Full reduction proofs are provided in appendix \ref{app:formal_proofs}. In section \ref{clinconnet2} we present our effective instantiation and use-case: blockchain-based consent management for clinical trials. Finally in Section \ref{concl} we conclude the paper.

\section{Related works}
\label{background}
Ring signatures (RS) originally provided unconditional anonymity for ad-hoc groups without requiring a group manager nor group coordination \cite{rivest}. LSAG \cite{LRS} addressed the inability to detect double-signing, extending the use of ring signatures to e-voting applications. LSAG linkability was global and depend on the ring, so the two signatures from the same signer were always linkable within the same ring. MLSAG \cite{MLSAG} advanced this scheme by using scopes to break the global linkability and the ring dependence, yielding signatures that are only linkable in a specific scope and that were used for anonymous transactions on blockchain \cite{monero}. These linkable RS variants do not provide any accountability for the signers. 

Fujisaki and Suzuki \cite{TRS} introduced Traceable RS (TRS). TRS is linkable and uses a 'tag' (referred to as scope in this paper) that allows the signer's identity to be revealed mathematically if they sign twice within the same tag and ring. This effectively provides a form of conditional accountability with scoped linkability. Xu and Yung \cite{Xu-acc} proposed Accountable RS (ARS) that enabled revocation of anonymity via designated trusted opener - or a threshold of openers - but lacked linkability and relies on a trusted dealer to generate and distribute keys. 

Bootle et al.\cite{Boot} introduced Short Accountable RS (SARS), in which the signer must encrypt their index within the ring using the public key of an opener and generate a non-interactive zero-knowledge (NIZK) proof to demonstrate it. The NIZK proof is a separate commitment from the signature, and SARS offered no linkability. 

RS in principle removed the centralization and the need for a group manager traditionally present in group signatures, so introducing a centralized accountability authority would be counter-productive. A clear literature gap in RS schemes is identified: no existing scheme natively integrates scoped linkability and on-demand decentralized accountability. Table \ref{tab:comparison} summarizes the limitations of existing RS schemes.

\begin{table}[ht]
\centering
\caption{Comparison of linkable ring signatures}
\label{tab:comparison}
\resizebox{\textwidth}{!}{%
\begin{threeparttable}
\begin{tabular}{|l |l |l |l |l |l|}
\hline
\textbf{Schemes} & \textbf{Linkability} & \textbf{Ring-Independence} & \textbf{Accountability} & \textbf{Opener} & \textbf{Separate Commitment} \\
\hline
RS \cite{rivest} & None & N/A\tnote{} & None & N/A & N/A \\ \hline
LSAG \cite{LRS} & Global & No & None & N/A & N/A \\ \hline
MLSAG \cite{MLSAG} & Scoped & Yes & None & N/A & N/A \\ \hline
ARS \cite{Xu-acc} & None & N/A & On-Demand & Designated / Dealer & N/A \\ \hline
TRS \cite{TRS} & Scoped & No & Conditional & Mathematical Extraction & No (Embedded) \\ \hline
SARS \cite{Boot} & None & N/A & On-Demand & Centralized Authority & Yes (External NIZK) \\ \hline
\textbf{DSLRS (Ours)} & \textbf{Scoped} & \textbf{Yes} & \textbf{On-Demand} & \textbf{Decentralized (Threshold)} & \textbf{No (Embedded)} \\ \hline
\end{tabular}
\begin{tablenotes}
\item[] N/A for 'Not Applicable'.
\end{tablenotes}
\end{threeparttable}%
}
\end{table}

\section{Defining DSLRS signature scheme}
\label{sec3}
This section formalizes the DSLRS scheme. It specifies the used cryptographic primitives, the threat model, the security assumptions and details the DSLRS algorithms.
A comprehensive notation reference table is provided in Table \ref{tab:notations}. \\

\subsection{Overview}
\label{overview}
The DSLRS scheme constructs a circular challenge loop denoted as $(ch_i)$ using Fiat-Shamir heuristic for each user of the ring $L$. Challenges for non-signers use the uniformly random responses $(x_i,z_i) \xleftarrow[]{R}\mathbb Zq$ while the signer computes their response $(x_s,z_s)$ using their secret key to mathematically close the loop. Verification requires a single challenge (provided in the signature, by default $ch_1$) and the responses $\{x_i,z_i\}_{i=1}^n$ to independently compute challenges and verify the signature.

Linkability uses a "scope" to separate signatures into different contexts, effectively providing scoped linkability and cross-scope unlinkability. This linkability feature is independent from the used ring, as it is based on the scope and the signer's public key.

Accountability is achieved by embedding an ElGamal-encrypted deanonymization tuple $(C_1,C_2)$ directly into the challenges. Moreover, the extraction of the signer's identity requires a threshold from a decentralized deanonymization network, and it extracts the public key instead of their index in the ring, making DSLRS ring-independent for accountability too. The deanonymization network runs a Distributed Key Generation (DKG) protocol \cite{Gennaro,Pederson} to generate necessary parameters to achieve decentralized accountability.

The security of this scheme and its properties are formally proven under \textbf{the threat model} described below and the \textbf{security assumptions} in Section \ref{preliminaries}.

\textbf{Threat model:}
We assume a Probabilistic Polynomial-Time (PPT) adversary $\mathcal{A}$ that has access to public parameters $PP$ (cf. Section \ref{scheme-def}). We assume $\mathcal{A}$ cannot alter the established $PP$. $\mathcal{A}$ is capable of making up to a polynomial number of queries $(q_H)$ to the random hash oracles $H$ and $H_p$ and $(q_{O_s})$ queries the signing oracle $O_s$.

We assume an honest majority in the decentralized deanonymization network composed of $N$ independent nodes. The maximum number of corruptible nodes is $c < (k-1)$ where $k$ is the network's secret key reconstruction threshold (cf. $NetKeyGen$ in Section \ref{scheme-def}).

\subsection{Preliminaries, notations and assumptions}
\label{preliminaries}
Let $p$ be a large prime number, $\mathbb F_p$ the finite field of order $p$, and $E$ an elliptic curve over $\mathbb F_p$ denoted as $E(\mathbb F_p)$. $E(\mathbb F_p)$ is defined by the simplified Weierstrass equation as $y^2 = x^3 + ax + b \pmod p$ where $a,b \in \mathbb F_p$ and $4a^3 + 27b^2 \neq 0 \pmod p$ \cite{SEC1}. 
Let $G$ be an element of a large prime order $q$ of $E(\mathbb F_p)$. $G$ generates a cyclic subgroup $\mathbb{G}$ of $E(\mathbb F_p)$, such that 
$\mathbb{G} = \{k \cdot G \mid k \in \mathbb{Z}_q\}$ and $q \cdot G = \mathcal{O}$, where $\mathcal{O}$ is the identity element of $\mathbb G$.

The \textbf{Elliptic Curve Discrete Logarithm Problem (ECDLP)} assumption holds relative to $\mathbb G$ if for any PPT adversary $\mathcal{A}$, the probability of finding $x \in \mathbb Zq$ that satisfies $Y = x \cdot G$ for $Y \in \mathbb G$ is negligible. We note the advantage of $\mathcal{A}$ in solving the ECDLP as
\\
$Adv^{ECDLP}_{\mathcal{A}} = Pr[\mathcal{A}(G, x\cdot G)=x | x \xleftarrow{R} \mathbb Zq] \approx 0$.

The \textbf{Decisional Diffie-Hellman (DDH)} assumption also holds relative to $\mathbb G$ if for any PPT adversary $\mathcal{A}$, the probability of distinguishing a Diffie-Hellman tuple $T_{DH} = (G, a\cdot G, b \cdot G, ab \cdot G)$ from a random tuple $T_{random} = (G, a \cdot G, b \cdot G, c \cdot G)$ for uniform scalars $a,b,c \xleftarrow{R} \mathbb Zq$ is negligible. We note the advantage of $\mathcal{A}$ in solving DDH as
\\
$Adv^{DDH}_{\mathcal{A}}=| Pr[a,b,c \xleftarrow{ }\mathbb Zq, \beta \xleftarrow{R} \{0,1\}: \mathcal{A}(G, a \cdot G, b \cdot G, ((1-\beta)ab + \beta c) \cdot G) = \beta] - \frac{1}{2}| \approx 0$. 

\textbf{The two hash functions $H$ and $H_p$} are modeled in the \textbf{Random Oracle Model} (ROM). 
\begin{table}
\footnotesize
\caption{Notations}
\label{tab:notations}
\centering
\begin{tabularx}{\textwidth}{|p{0.3\textwidth}|X|}
\hline
\textbf{Notations} & \textbf{Description} \\ \hline
$\mathbb G, G, q$ & The EC group $\mathbb G$ generated from $G$ and of prime order $q$. \\ \hline
$(S_i, P_i)$ & A key pair of user $i$, $i \in \{1, \dots, K\}$, where $S_i \in \mathbb Zq^*$ is user $i$'s secret key (scalar) and $P_i \in \mathbb L$ is the public key computed as $P_i = S_i \cdot G$. Let $s$ be the index of the signer and $(P_s,S_s)$ their key pair. \\ \hline
$\mathbb L = \{P_1, P_2, \dots, P_K\}$ & Global public key registry of $K$ users, $K \geq n_{min}$ where $n_{min}$ is a fixed number representing the smallest possible ring. \\ \hline
$L = \{P_1, P_2, \dots, P_n\}$ & A ring of $n$ public keys which are sampled by the signer from $\mathbb L$, $P_s \in L$, and $ n_{min} \leq n \leq K$. \\ \hline
$(S_{net},P_{net}), \{S_{net-j}\}_{j=1}^N$ & $(S_{net},P_{net})$ is the key pair of the deanonymization network of $N$ nodes. \\
$\{\omega_j\}_{j=1}^N, \{\lambda_j\}_{j=1}^N$ & Each participating node $j$ has: a public index $\omega_j$, a coefficient $\lambda_j$ and a Shamir secret key share $S_{net-j}$, where $\lambda_j = \prod_{i=1,i \neq j}^k \frac{\omega_i}{\omega_i - \omega_j} \pmod q$, $S_{net} = \sum_{j=1}^{k} \lambda_j \cdot S_{net-j} \pmod q$ and $P_{net} = S_{net}\cdot G$.\\ \hline
$SID \xleftarrow{} \mathbb Z_q$ & Identifier of a scope. \\ \hline
$I_{scope} \in \mathbb G$ & Key image of the key pair $(S_s, P_s)$ of the signer for a given scope $SID$, computed as $I_{scope} = S_s \cdot H_p(P_s || SID)$. \\ \hline
$m = \{0,1\}^*$ & The message to be signed. \\ \hline
$\sigma$ & The digital signature issued by the signer. \\ \hline
$a \xleftarrow{R} \mathbb Zq$ & $a$ is the uniform random draw of a value from $\mathbb Zq$ \\ \hline
$H: \{0,1\}^* \xrightarrow{} \mathbb Zq^*$ & A hash function that maps an arbitrary-length binary input string to a random scalar in $\mathbb Z_q^*$ (in the Random Oracle Model). \\ \hline
$H_p: \{0,1\}^* \xrightarrow{} \mathbb G$ & A hash-to-point function that maps an arbitrary-length binary input string to a random valid point on the subgroup $\mathbb G$ (in the Random Oracle Model). \\ \hline
$O_s(m,P_i,L,SID,PP) \xrightarrow{} \sigma$ & A \textbf{signing Oracle} that receives a message $m$, a public key $P_{i} \in L$, a ring $L$, a scope identifier $SID$, and $PP$, returns a valid signature $\sigma$ generated using the secret key $S_i$. 
\\ \hline
\end{tabularx}
\end{table}

\subsection{DSLRS scheme definition}
\label{scheme-def}
A DSLRS scheme over a PPT setup $Setup$ is a tuple of polynomial-time algorithms $(Setup, Sign, Verify, Link, Deanonymize)$ defined as follows:

\begin{itemize}
    \item $Setup(1^\lambda, N, K, P) \xrightarrow{} (PP = (\mathbb G, G, q, H, H_p, \{SID_i\}_{i=1}^P, P_{net}, \{\omega_j\}_{j=1}^N,\mathbb L=\{P_i\}_{i=1}^{K},n_{min}), \{S_i\}_{i=1}^{K}, \{S_{net-j}\}_{j=1}^{N})$: \\
    Given the security parameter $\lambda$, $N$ deanonymization network (opener) nodes, $K$ users, and $P$ contexts, the full system generates public parameters $PP$, users' public keys $\mathbb L=\{P_i\}_{i=1}^{K}$ and private keys $\{S_i\}_{i=1}^{K}$ and the network's public key $P_{net}$ and secret key shares $\{S_{net-j}\}_{j=1}^{N}$ and public indices $\{\omega_j\}_{j=1}^N$, as follows:
    \begin{enumerate}
        \item Generate base parameters $\{\mathbb G, G, q, H, H_p,\{SID_i\}_{i=1}^P, n_{min} \}$ where  authorized in the system.
        \item Execute [$NetKeyGen(PP) \xrightarrow{} (\{S_{net-j}\}_{j=1}^N,P_{net})$]. The DKG aggregates the network's public key $P_{net}$ using individual secret key shares $S_{net-j}$ from each node $j$ in the network. A node $j$ has a public index $\omega_j \in PP$ assigned by $NetKeyGen$. The secret key of the deanonymization network is reconstructed as $S_{net} = \sum_{j=1}^{k} \lambda_j \cdot S_{net-j} \pmod q$ using $k$ shares (k-of-N threshold) \cite{Shamir}, where $\lambda_j = \prod_{i=1,i \neq j}^k \frac{\omega_i}{\omega_i - \omega_j} \pmod q$. $P_{net} = S_{net}\cdot G$.
        \item Execute $K$ times [$Gen(PP) \xrightarrow{} (S_i,P_i)$] for $K$ users. $Gen(PP)$ issues a valid pair $(P_i,S_i)$; $P_i \in \mathbb G$ and $S_i \in \mathbb Z_q^*$ for user $i$, where $P_i = S_i \cdot G$. This initializes the global public key registry $\mathbb L = \{P_1, \cdots, P_K\}$ \footnote{This registry can be dynamically updated with a Decentralized Key Registration mechanism as presented in section \ref{clinconnet2}.} and associated secret keys $\{S_{j}\}_{j=1}^{K}$.
    \end{enumerate}
    \item $Sign(m,S_s,L,SID,PP) \xrightarrow{} \sigma = (I_{scope}, L, SID, P_{net}, C_1, C_2, ch_1, \{x_i, z_i\}_{i=1}^n)$: Given a message $m$, a ring $L= \{P_1, \cdots, P_n\}$ of $n \geq n_{min}$ public keys sampled by the signer from $\mathbb L \in PP$, the secret key $S_s$ of the signer associated with his public key $P_s \in L$, a scope $SID$ and $PP$, returns either a computed signature $\sigma$ or $\bot$ if the inputs are malformed or $P_s \notin L$. This algorithm computes commitments, responses $\{x_i, z_i\}_{i=1}^n$ and challenges $\{ch_i\}_{i=1}^n$ for all users in the ring $L$.
    \item $Verify(m,\sigma,PP) \xrightarrow{} 1/0$: Given a message $m$, a signature $\sigma$ and public parameters $PP$, returns 1 if the signature is valid and 0 otherwise. $Verify$ uses the input signature containing responses $\{x_i, z_i\}_{i=1}^n$ to independently rebuild all challenges $\{ch_i\}_{i=1}^n$. Since challenges are a circular loop that can be verified at any point, this algorithm by default verifies at the index $i=1$.
    \item $Link(\sigma_1, m_1,\sigma_2,m_2,PP) \xrightarrow{} 1/0$: Given two signatures $\sigma_1$ and $\sigma_2$, two messages $m_1$ and $m_2$, and $PP$, returns 1 if the signatures share the same signer in the same scope and 0 otherwise. After running $Verify$ on both signatures, this algorithm uses the key images $I_{scope}$ of both signatures to verify if the signatures were made by the same signer in this scope or not.
    \item $Deanonymize(\sigma, \{D_j\}_{j=1}^k, PP) \xrightarrow{} P_s$: Given a signature $\sigma$,  the calculated transient decryption shares $D_j = S_{net-j} \cdot C_1$ from k of $N$ deanonymization network's nodes and their public indices $\{\omega_j\}_{j=1}^k \in PP$, returns the public key of the signer $P_s$. $Deanonymize$ allows the extraction of $P_s$ from the committed $(C_1,C_2)$ in the signature $\sigma$ using $PP$ and the decryption shares of $k$ nodes.
\end{itemize}

\noindent
\begin{minipage}[t]{0.48\textwidth}
\begin{algorithm}[H]
\footnotesize
\caption{Sign}
\begin{algorithmic}[1]
\State \textbf{Inputs:} $m, S_s,L, SID,PP$
\If {$(P_s \not \in L) \lor (L \not \subset \mathbb L) \lor (SID \not \in \{SID\}_{i=1}^P) \lor (length(L) < n_{min})$}
    \State \Return $\bot$
\EndIf
\State $I_{scope} = S_s \cdot H_p(P_s || SID)$
\State $r, r_{dean}, r_z \xleftarrow{R} \mathbb Zq$
\State $C_1 = r_{dean} \cdot G$ ; $C_2 = P_s + r_{dean} \cdot P_{net}$
\State $L_s = r \cdot G$ ; $R_s = r \cdot H_p(P_s || SID)$
\State $A_s = r_z \cdot G$ ; $B_s = r_z \cdot P_{net}$
\State $ch_{(s \pmod n)+1} = H(m, L, I_{scope}, SID, C_1, C_2, L_s, R_s, A_s, B_s)$ 
\For{$j=1$ to $n-1, i \gets ((s+j-1) \pmod n) + 1)$}
    \State $x_i, z_i \xleftarrow{R} \mathbb Zq$
    \State $L_i = x_i \cdot G + ch_i \cdot P_i$
    \State $R_i = x_i \cdot H_p(P_i || SID) + ch_i \cdot I_{scope}$
    \State $A_i = z_i \cdot G - ch_i \cdot C_1$
    \State $B_i = z_i \cdot P_{net} - ch_i \cdot (C_2 - P_i)$
    \State $ch_{(i \pmod n)+1} = H(m, L, I_{scope}, SID, C_1, C_2, L_i, R_i, A_i, B_i)$
\EndFor
\State $x_s = r - ch_s \cdot S_s \pmod q$
\State $z_s = r_z + ch_s \cdot r_{dean} \pmod q$
\State $\sigma \gets (I_{scope}, L, SID, P_{net},C_1, C_2, ch_1,$\\
$\{x_i, z_i\}_{i=1}^n)$
\State \Return $\sigma$
\end{algorithmic}
\end{algorithm}
\end{minipage}
\hfill
\begin{minipage}[t]{0.48\textwidth}
\begin{algorithm}[H]
\footnotesize
\caption{Verify}
\begin{algorithmic}[1]
\State \textbf{Inputs:} $m$, $\sigma = (I_{scope}, L, SID, P_{net},C_1, C_2, ch_1, \{x_i, z_i\}_{i=1}^n), PP$
\If {$(L \not \subset \mathbb L) \lor (SID \not \in \{SID\}_{i=1}^P)) \lor (P_{net} \not \in PP) \lor (length(L) < n_{min} \in PP)$}
    \State \Return 0
\EndIf
\If {($(I_{scope} \notin \mathbb G) \lor (C_1 \notin \mathbb G) \lor (C_2 \notin \mathbb G)$}
    \State \Return $0$.
\EndIf
    \State $L_1' = x_1 \cdot G + ch_1 \cdot P_1$
    \State $R_1' = x_1 \cdot H_p(P_1 || SID) + ch_1 \cdot I_{scope}$
    \State $A_1' = z_1 \cdot G - ch_1 \cdot C_1$
    \State $B_1' = z_1 \cdot P_{net} - ch_1 \cdot (C_2 - P_1)$
    \State $ch'_2 = H(m, L, I_{scope}, SID, C_1, C_2, L_1', R_1', A_1', B_1')$

\For{$i = 2$ to $n$}
    \State $L_i' = x_i \cdot G + ch'_i \cdot P_i$
    \State $R_i' = x_i \cdot H_p(P_i || SID) + ch'_i \cdot I_{scope}$
    \State $A_i' = z_i \cdot G - ch'_i \cdot C_1$
    \State $B_i' = z_i \cdot P_{net} - ch'_i \cdot (C_2 - P_i)$
    \State $ch'_{(i \mod n) + 1} = H(m, L, I_{scope}, SID, C_1, C_2, L_i', R_i', A_i', B_i')$
\EndFor
\If{$ch_1' == ch_1$}
    \State \Return $1$
\EndIf
\State \Return $0$
\end{algorithmic}
\end{algorithm}
\end{minipage}
\hfill
\begin{minipage}[t]{0.48\textwidth}
\begin{algorithm}[H]
\footnotesize
\caption{Link}
\begin{algorithmic}[1]
\State \textbf{Inputs:} $\sigma_1$, $m_1$, $\sigma_2$, $m_2$, $PP$
\If {($SID_{\sigma_1} \neq SID_{\sigma_2}$)}
    \State \Return $0$
\EndIf
\If {($Verify(m_1,\sigma_1,PP) \land Verify(m_2,\sigma_2,PP)$)}
    \If{($I_{scope}^{\sigma_1} == I_{scope}^{\sigma_2}$)}
    \State \Return $1$
    \EndIf
\EndIf
\State \Return $0$
\end{algorithmic}
\end{algorithm}
\end{minipage}
\hfill
\begin{minipage}[t]{0.48\textwidth}
\begin{algorithm}[H]
\footnotesize
\caption{Deanonymize}
\begin{algorithmic}[1]
\State \textbf{Inputs:} $\sigma = (I_{scope}, L, SID, P_{net},C_1, C_2, ch_1, \{x_i, z_i\}_{i=1}^n)$, $\{D_j = S_{net-j}\cdot C_1\}_{j=1}^k$, $PP$
\For {$j = 1$ to $k$}
    \State $\lambda_j = \prod_{i=1,i \neq j}^k \frac{\omega_i}{\omega_i - \omega_j} \pmod q$
\EndFor
\State $V = \sum_{j=1}^{k} \lambda_j \cdot D_j$ 
\State $P_s = C_2 - V$ 
\State \Return $P_s$
\end{algorithmic}
\end{algorithm}
\end{minipage}

\subsection{DSLRS signature evaluation and performance}
\label{eval}
Let $|\mathbb G|$ be the size of a point in $\mathbb G$ and $|\mathbb Z_q|$ the size of a scalar.\\
The size of a DSLRS signature $\sigma = (I_{scope}, L, SID, P_{net},C_1, C_2, ch_1, \{x_i, z_i\}_{i=1}^n)$ can be evaluated as follows:
$|\sigma| = (n+4) \cdot |\mathbb G| + (2n+2) \cdot |\mathbb Z_q|$

The spatial evaluation is linear, $|\sigma| = \mathcal{O}(n)$. Assuming a 256-bit curve, the compressed EC point size is $|\mathbb G| = 33$ bytes and $|\mathbb Z_q| = 32$ bytes. The signature size is $|\sigma| = 97n + 196$ bytes, equivalent to approximately $0.97$ KB for a ring of size $n=8$, $1.74$ KB for $n=16$, and $3.3$ KB for $n=32$.
The algorithms $\{Sign,Verify,Link\}$ have an $\mathcal{O}(n)$ computational complexity since they all involve constructing or reconstructing challenges for all ring members. Compared to related works like RS \cite{rivest}, LSAG \cite{LRS}, MLAG \cite{MLSAG}, ARS \cite{Xu-acc} and TRS \cite{TRS}, the complexity of DSLRS aligns with the baseline of linear complexity $\mathcal{O}(n)$ for standard ring signatures. While SARS \cite{Boot} has a sub-linear complexity $\mathcal{O}(\log_{2} n)$, making DSLRS less efficient in both spatial and computational complexity if compared to SARS, SARS requires the verification of a separate NIZK proof for accountability, which may also have a sub-linear complexity. Nevertheless, in DSLRS, the $Deanonymize$ function is ring-independent and only depends on the number of the threshold nodes $k$, making its complexity an $\mathcal{O}(k)$, which is an $\mathcal{O}(1)$ relative to $n$. In terms of accountability, the computational complexity and self-sufficient nature of DSLRS make it more practical than SARS.
\\
A simple implementation of the 4 DSLRS algorithms in Python yields the following mean and median latencies and signature sizes. For this implementation, we chose secp256k1 (ECDSA) as the elliptic curve, with a deanonymization network of $N=12$ members and a threshhold of $k=8$.
We also use 3 ring sizes $n=8, \quad n=16$ and $n=32$. Table \ref{tab:results} summarizes the experimental results. These numbers were obtained on a Dell Precision 5480, with the 13th Gen Intel CPU core i9-13900Hx20 running 20 Cores at 1.39 Ghz. The machine has 32 GB of RAM. The operating system is Ubuntu 24.04.1 LTS of 64-bit.
For each algorithm, 100 iterations are run per ring size.

\begin{table}[h]
    \centering
    \begin{tabular}{|c | c | c | c|}
    \hline
    Ring size & $n=8$ & $n=16$ & $n=32$ \\ 
    Value & (Mean / Median) & (Mean / Median) & (Mean / Median) \\ \hline
    Sign & 270.6 ms/ 266.2 ms & 548.0 ms/ 543.4 ms & 1133.6 ms/ 1128.2 ms \\ \hline
    Verify & 269.8 ms/ 268.4 ms & 542.8 ms/ 542.5 ms & 1088.3 ms / 1085.0 ms \\ \hline
    Link & 542.0 ms/ 537.3 ms & 1085.3 ms/ 1077.7 ms & 2189.0 ms / 2170.2 ms \\ \hline
    Deanonymize & 15.6 ms/ 15.8 ms & 15.3 ms/ 15.2 ms & 15.8 ms/ 16 ms\\ \hline
    Signature size & 972 bytes & 1748 bytes & 3300 bytes\\ \hline
    \end{tabular}
    \caption{DSLRS implementation results showing mean and median latencies of the algorithms, as well as the signature size per ring size from 100 iterations per ring size.}
    \label{tab:results}
\end{table}
As Table \ref{tab:results} shows, the results are consistent with the theoretical analysis: $\{Sign,Verify,Link\}$ are $\mathcal{O}(n)$ and $Deanonymize$ is $\mathcal{O}(1)$.
Figure \ref{fig:bench} shows the distribution of results around the median value (the horizontal line in the middle), with the mean value represented by a white dot. Note that the quartiles are very close to the median line, which indicates consistent latency and algorithm execution.
\begin{figure}
    \centering
    \includegraphics[width=1\linewidth]{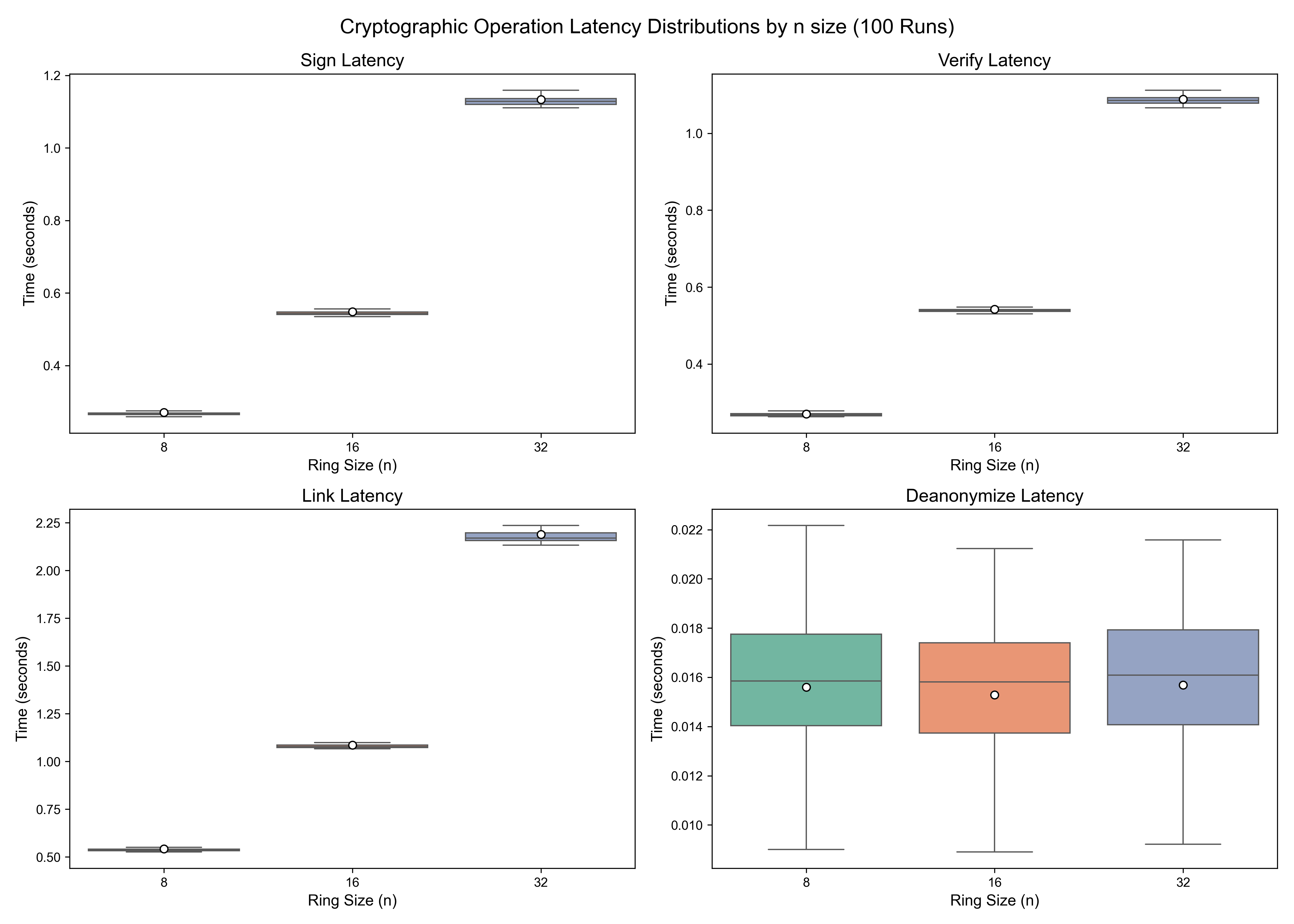}
    \caption{Distribution of the latency of DSLRS algorithms by ring size. The white dots represent the mean values, and the horizontal lines show the median value (in the middle) and the quartile range.}
    \label{fig:bench}
\end{figure}

\section{Security analysis and formal proofs}
\label{analysis}

\begin{theorem}
The DSLRS scheme is perfectly correct, satisfies signer indistinguishability, EUF-CMA unforgeability, scoped linkability, cross-scope unlinkability and accountability under the ECDLP and DDH hardness assumptions in the Random Oracle Model (ROM) (cf. Section \ref{preliminaries}) and the threat model presented in Section \ref{overview}.
\end{theorem}
The full formal proofs are provided in appendix \ref{app:formal_proofs}. Intuitive simple proofs are provided under each lemma.

\begin{lemma}
\textbf{Perfect correctness:} DSLRS scheme satisfies perfect correctness if: 

\[
\Pr \left[
\begin{array}{c}
    (PP, \{S_i\}_{i=1}^{K}, \{S_{net-j}\}_{j=1}^{N}) \leftarrow{} \text{Setup}(1^{\lambda}, N,K,P); \\
    \sigma \leftarrow{} Sign(m,S_s,L,SID,PP); \\
    \left(
    \begin{array}{c}
    (\forall m \in \{0,1\}^*) \land (\forall SID \in \{SID_i\}_{i=1}^P) \land (\forall L \subseteq \mathbb{L}) \land (\forall (S_s,P_s)) \land (P_s \in L) \\
     Verify(m,\sigma,PP) = 1 \land  Deanonymize(\sigma, \{D_j\}_{j=1}^k, PP)=P_s \\

     \end{array}
    \right)
\end{array}
\right] = 1
\]
\end{lemma}
We prove this by demonstrating that an honest signer is able to produce a signature by closing the challenge loop computing $x_s,z_s$, and that the challenge loop can be reconstructed and closed during verification using a challenge (here $ch_1$ by default) and the provided commitments and responses. See \ref{app:proof_lemma_1} for formal proof.
\begin{lemma}
\textbf{Signer indistinguishability:} The DSLRS scheme satisfies signer indistinguishability if for any PPT adversary $\mathcal{A}$:\\ 
\[
Adv_{\mathcal{A}}^{indist} = \left| \Pr \left[
\begin{array}{c}
(PP, \{S_i\}_{i=1}^{K}, \{S_{net-j}\}_{j=1}^{N}) \leftarrow{} \text{Setup}(1^{\lambda}, N,K,P); b \xleftarrow{R} \{0,1\}; \\
    (m^*,SID^*,L^*,P_{j_0}^*,P_{j_1}^*) \leftarrow{} \mathcal{A}^{O_s}(PP); \\
    \sigma^* \leftarrow{} Sign(m^*,S_{j_b},L^*,SID^*,PP); \\
    b' \xleftarrow{} \mathcal{A}(\sigma^*) = b \\
    \left(
    \begin{array}{c}
    \land (L^* \subseteq \mathbb{L}) \land (\{P_{j_0}^*,P_{j_1}^*\} \subseteq L^* )\\
     \land ((P_{j_0}^*,SID^*),(P_{j_1}^*,SID^*) \text{ not queried to } O_s)
     \end{array}
     \right)
\end{array}
\right] - \frac{1}{2} \right| \approx 0
\]
\end{lemma}
We prove this in \ref{app:proof_lemma_2} by proving that all of the signature components are either uniformly distributed and random scalars in $\mathbb Z_q$ or uniformly distributed points of $\mathbb G$. Moreover, we prove that under the DDH hardness assumption, the deanonymization components $C_1,C_2$ offer no advantage in guessing the signer's identity.
\begin{lemma}
\textbf{Unforgeability --  EUF-CMA:} The DSLRS scheme is existentially unforgeable under chosen-message attacks (EUF-CMA) if for any PPT adversary $\mathcal{A}$:
\[
Adv_{\mathcal{A}}^{EUF-CMA}= \Pr \left[
\begin{array}{c}
    PP, \{S_i\}_{i=1}^{K}, \{S_{net-j}\}_{j=1}^{N} \leftarrow{} \text{Setup}(1^{\lambda}, N,K,P); \\
    (m^*,SID^*,L^*,\sigma^*) \leftarrow{} \mathcal{A}^{O_s,H,H_p}(PP); \\
    \left(
    \begin{array}{c}
    (L^* \subseteq \mathbb L) \land (\forall P_i \in L^* , S_i \text{ is unkown to } \mathcal{A}) \\
    \land Verify(m^*,\sigma^*,PP) = 1 \\
    \land ((m^*,SID^*,L^*) \text{ not queried to } O_s)
    \end{array}
    \right)

\end{array}
\right] \approx 0
\]
\end{lemma}
We use the forking lemma for ring signatures \cite{forking-lemma} to prove in \ref{app:proof_lemma_3} that if DSLRS is forgeable (EUF-CMA), a simulator $\mathcal{B}$ can utilize an adversary $\mathcal{A}$ with such forgery capacity to solve an instance of ECDLP, which contradicts the ECDLP hardness assumption.
\begin{corollary}
\textbf{Key image attack resistance:} The DSLRS unforgeability property holds even if $\mathcal{A}$ obtains a valid key image $I_{scope}$ generated by an honest signer $s$.
\end{corollary}
This corollary is proven by demonstrating that a valid DSLRS signature challenge loop can only be completed if the signer knows the private key associated to the used key image.
\begin{corollary}
\textbf{Key image integrity:} The DSLRS scheme guarantees that if a signature $\sigma$ is valid, the embedded key image $I_{scope}$ is strictly formed as $S_s \cdot H_p(P_s||SID)$.
\end{corollary}
Corollary 2 is proven by showing that a malformed key image cannot be used to produce a valid DSLRS signature.
\\
Full proofs for both corollaries are provided in \ref{app:proof_lemma_3}.
\begin{lemma}
\textbf{Scoped-linkability:} The DSLRS scheme satisfies scoped-linkability if for any PPT adversary $\mathcal{A}$:
\[
Adv_{\mathcal{A}}^{SL} = \Pr \left[
\begin{array}{c}
   PP, \{S_i\}_{i=1}^{K}, \{S_{net-j}\}_{j=1}^{N} \leftarrow{} \text{Setup}(1^{\lambda}, N,K,P); \\
   (\sigma_1^*,\sigma_2^*,m_1^*,m_2^*,SID^*,L_1^*,L_2^*,P_{s1}^*,P_{s2}^*) \leftarrow{} \mathcal{A}^{O_s,H,H_p}(PP);\\
   \left(
   \begin{array}{c}
   (P_{s1}^* \in L_1^*) \land (P_{s2}^* \in L_2^*) \\
   \land ((m_1^*, P_{s1}^*, L_1^*, SID^*), (m_2^*, P_{s2}^*,L_2^*, SID^*) \text{ not queried to } O_s) \\
   \land (Verify(m_1^*,\sigma_1^*,PP) = 1) \land (Verify(m_2^*,\sigma_2^*,PP) = 1) \\
   \land \left(
   \begin{array}{c}
   (P_{s1}^* = P_{s2}^*) \land Link(\sigma_1^*,m_1^*,\sigma_2^*,m_2^*,PP) = 0) \\
   \lor \\
   (P_{s1}^* \neq P_{s2}^*) \land Link(\sigma_1^*,m_1^*,\sigma_2^*,m_2^*,PP) = 1) \\ 
   \end{array}
   \right) \\
   \end{array}
   \right)

\end{array}
\right] \approx 0
\]
\end{lemma}
Lemma 4 is proven in \ref{app:proof_lemma_4} using corollary 2. The corollary contradicts the first case, while under the second case it shows that breaking the linkability requires solving ECDLP, which is assumed hard.
\begin{lemma}
\textbf{Cross-scope unlinkability:} The DSLRS scheme provides cross-scope unlinkability if for any PPT adversary $\mathcal{A}$:

\[
Adv_{\mathcal{A}}^{CSU} = \left| \Pr \left[
\begin{array}{c}
   PP, \{S_i\}_{i=1}^{K}, \{S_{net-j}\}_{j=1}^{N} \leftarrow{} \text{Setup}(1^{\lambda}, N,K,P); b \xleftarrow{R} \{0,1\}; \\
    (m^*,SID_0^*,SID_1^*,L^*,P_{j_0}^*,P_{j_1}^*) \leftarrow{} \mathcal{A}^{O_s,H,H_p}(PP); \\
    \sigma_1 \leftarrow{} Sign(m^*,S_{j_0},L^*,SID_0^*,PP); \\
    \sigma_2 \leftarrow{} Sign(m^*,S_{j_b},L^*,SID_1^*,PP); \\
    b' \xleftarrow{} \mathcal{A}(\sigma_1,\sigma_2) = b \\
    \left(
    \begin{array}{c}
    (SID_0^* \neq SID_1^*)\\
    \land (L^* \subseteq \mathbb{L}) \land (\{P_{j_0}^*,P_{j_1}^*\} \subseteq \mathbb L ) \\ \land ((P_{j_0}^*,SID_1^*),(P_{j_1}^*,SID_1^*) \text{ not queried to } O_s)
    \end{array}
    \right)
\end{array}
\right] - \frac{1}{2} \right| \approx 0
\]
\end{lemma}
We prove in \ref{app:proof_lemma_5} that if DSLRS is linkable cross-scope, a simulator $\mathcal{B}$ can utilize an adversary $\mathcal{A}$ with such linking capability to solve an instance of DDH, which is assumed hard.
\begin{lemma}
\textbf{Accountability:} The DSLRS scheme satisfies accountability if for any PPT adversary $\mathcal{A}$:

\[
Adv_{\mathcal{A}}^{ACC} = \Pr \left[
\begin{array}{c}
    PP, \{S_i\}_{i=1}^{K}, \{S_{net-j}\}_{j=1}^{N} \leftarrow{} \text{Setup}(1^{\lambda}, N,K,P); \\
    (m^*,SID^*,L^*,\sigma^*) \leftarrow{} \mathcal{A}^{O_s}(PP,S_s); \\

    \left(
    \begin{array}{c}
    (Verify(m^*,\sigma^*,PP) = 1) \land (Deanonymize(\sigma^*,\{D_j\}_{j-1}^k)=P_{fake}^*) \\
    \land (P_{fake}^* \neq P_s)
    \end{array}
    \right)
\end{array}
\right] \approx 0
\]
\end{lemma}
In \ref{app:proof_lemma_6}, we prove that a fake identity $P_{fake}$ cannot be used in $C2$ to generate a valid signature. Since DSLRS is EUF-CMA (Lemma 3), producing a valid signature requires the deanonymization component $C_2$ to encapsulate the public key of the signer $P_s$.
\begin{corollary}
\textbf{Non-frameability:} The DSLRS scheme satisfies non-frameability if for any PPT adversary $\mathcal{A}$:
\[
Adv_{\mathcal{A}}^{FRAME} = \Pr \left[
\begin{array}{c}
    PP, \{S_i\}_{i=1}^{K}, \{S_{net-j}\}_{j=1}^{N} \leftarrow{} \text{Setup}(1^{\lambda}, N,K,P); \\
    (m^*,SID^*,L^*,\sigma^*,P_v^*) \leftarrow{} \mathcal{A}^{O_s,H,H_p}(PP,S_s); \\
    \left(
    \begin{array}{c}
    (P_v^* \in L^*) \land (P_s \in L^*) \land (P_s \neq P_v^*) \land (Verify(m^*,\sigma^*,PP) = 1) \\
    \land (Deanonymize(\sigma^*,\{D_j\}_{j-1}^k)=P_v^*)\\
    \land (m^*,P_v^*,SID^*) \text{ not queried to } O_s
    \end{array}

    \right)
\end{array}
\right] \approx 0
\]
\end{corollary}
From Lemma 6, corollary 3 is easily proven, since framing someone requires setting $C_2$ to the victim's public key $P_v$, which violates accountability. A full proof of the corollary is provided in \ref{app:proof_lemma_6}.

\section{Blockchain-based Consent Management using DSLRS signatures}
\label{clinconnet2}
We leverage a consortium blockchain $B$ (Hyperledger Fabric) maintained by $N$ \textbf{research organizations} ($RO$s) hereby referred to as "decentralized deanonymization network". Application of DKG on a blockchain \cite{B-DKG} facilitates the instantiation of the DSLRS signature scheme. We suppose that there is an honest majority of $RO$s where fewer than $k$ nodes are corrupted, $k$ being the DSLRS deanonymization threshold.

\subsection{DSLRS blockchain instantiation}
\label{setup}
The initialization of the signature system proceeds as follows:

\begin{enumerate}
    \item \textbf{Parameter Agreement:} The network nodes agree upon and immutably publish on the ledger the public parameters $PP$. \\
    $PP =\{\mathbb G, G, q, H, H_p,\{SID_i\}_{i=1}^P, n_{min} \}$. Each $SID_i$ uniquely identifies a research project. $SID_1$ is reserved for key registration.

    \item \textbf{Network Key Generation:} the $N$ nodes run $NetKeyGen$. $P_{net}$ and public indices of the nodes $\{\omega_j\}_{j=1}^N$ are published on the ledger and appended to $PP$. Each one of the $N$ nodes securely retains an individual secret key share $S_{net-j}$.
    
    \item \textbf{Smart Contracts Deployment:} Network nodes agree on publishing the smart contract "consent management contract" that uses DSLRS scheme to provide consent management functions as described in \ref{dcm}. An additional smart contract is deployed to enable \textbf{user registration} which permits users to publish their public keys.
    
    \item \textbf{Decentralized Public Key Registration and Proof of Possession (PoP):} Users (patients/participants) generate key pairs $(S_i,P_i)$ off-chain. To register $P_i$ to the global registry $\mathbb L$, they submit a Schnorr-based NIZK Proof of possession $\pi_i = SchnorrSign(S_i,m_i)$ made over $m_i = H(P_i||SID_1)$ to the user registration contract. The contract verifies that $P_i \in \mathbb G \land P_i \neq \mathcal{O} \land P_i \notin \mathbb L$ and that $SchnorrVerify(P_i,m_i,\pi_i) = 1$ before publishing $P_i$ on the blockchain.
    \begin{lemma}
    \textbf{Rogue Key Resistance:} A PPT adversary $\mathcal{A}$ has a negligible probability in registering a rogue key $P_{rogue}$.
    \end{lemma} Proof in \ref{app:proof_lemma_7}.
\end{enumerate}
This instantiation allows users to dynamically update and add public keys to $\mathbb L$ global registry. Once $K \geq n_{min}$, we can start our decentralized consent management system.

\subsection{Dynamic Consent Management}
\label{dcm}
We assume consent agreements are negotiated off-chain. Participants generate \textbf{consent proofs} $m_i$, defined as hashes of the concluded consent agreements. The \textbf{Consent Management Contract} provides a set of functions $\{Publish, Revoke, Deanonymize\}$ for participants. A consent proof on-chain has three states: $VALID$, $REVOKED$ and $REVEALED$.

A registered participant samples $L$ a ring of $n-1$ public keys from $\mathbb L$ and appends their registered public $P_s$ to it before executing $Sign(m,S_s,L,SID,PP)$ to generate a DSLRS signature $\sigma$. The participant submits $(m,\sigma)$ to the $Publish$ function.
The $Publish$ function executes $Verify(m,\sigma,PP)$ and if it passes, it writes $(m,\sigma)$ with the status $VALID$.

To revoke $m$, a participant executes $Sign('REV',S_s,L,SID,PP)$ to get $\sigma'$. They submit $('REV',\sigma')$ to $Revoke$. $Revoke$ runs $Link(\sigma,m,\sigma','REV')$ and if it passes, it changes the status of $(m,\sigma)$ to $REVOKED$ and writes the new $('REV',\sigma')$ as $VALID$. If a user wishes to instead update their consent with a new consent agreement $m'$, they submit $(m',\sigma')$ to $Revoke$ which effectively revokes and replaces the old consent.

Upon legal or clinical demand, $k$ nodes submit their computed transient decryption shares $\{D_j = S_{net-j} \cdot C_1\}_{j=1}^k$ to the consent management contract. It invokes $Deanonymize$ on  $(m,\sigma)$, transitions its state to $REVEALED$ and publishes the extracted signer's identity $P_s$.
Coupling this system with an external certification authority, $P_s$ is easily linkable to the legal identity of the patient.

\section{Conclusions}
\label{concl}
The DSLRS scheme solves a literature gap by providing a signature scheme with scoped linkability and accountability features. By embedding cryptographic commitments computed using ElGamal encrypted components alongside elements like dynamic key images that change from scope to scope and a network public key resulting from a DKG, the DSLRS signature challenges natively provide decentralized accountability without needing a separate commitment or relying on a trusted opener. Formal security reductions prove that under the ECDLP and DDH assumptions in the ROM model, DSLRS scheme satisfies perfect correctness, signer indistinguishability, EUF-CMA, scoped-linkability, cross-scope unlinkability and accountability. The proposed blockchain instantiation confirms practical use of our proposed system, especially within the proposed use-case which is consent management in clinical trials.

\section*{Data availability}
We provide the implemented DSLRS algorithms and the performance results in open source \url{https://github.com/montassar-isbored/DSLRS-simple}.
Google Gemini Pro 3 (AI model) was used to help generate the code and the benchmarking scripts.

\section*{Funding and Acknowledgments}
This work benefited from State aid managed by the
Agence Nationale de la Recherche (ANR) under the France
2030 programme, reference ANR-22-PESN-0006 (Project
TRACIA). It is also partly supported by the Chair Values and Policies of Personal Information (VPIP), Institut Mines-Telecom,  France, and International Alliance for
Strengthening Cybersecurity and Privacy in Healthcare (CybAlliance, Project no. 337316).

\bibliographystyle{unsrt}

\appendix

\section{Formal proofs}
\label{app:formal_proofs}

\subsection{Proof of Lemma 1 (Perfect correctness)}
\label{app:proof_lemma_1}
\textit{Proof for Verify:}  An honest signer with index $s$ in ring $L$ produces $x_s,z_s \in Z_q$ such that the initial commitments can be successfully reconstructed during verification.
$x_s,z_z$ are computed by the signer as $x_s = r - ch_s \cdot S_s \pmod q$ and $z_s = r_z + ch_s \cdot r_{dean} \pmod q$. The signer commitments during verification are:
\begin{itemize}
    \item $L_s' = x_s \cdot G + ch'_s \cdot P_s = (r - ch_s \cdot S_s) \cdot G + ch'_s \cdot (S_s \cdot G) = r \cdot G = L_s$
    \item $R_s' = x_s \cdot H_p(P_s || SID) + ch'_s \cdot I_{scope} = (r - ch_s \cdot S_s) \cdot H_p(P_s || SID) + ch'_s \cdot (S_s \cdot H_p(P_s || SID)) = r \cdot H_p(P_s || SID) = R_s$
    \item $A_s' = z_s \cdot G - ch'_s \cdot C_1 = (r_z + ch_s \cdot r_{dean}) \cdot G - ch_s \cdot (r_{dean} \cdot G) = r_z \cdot G = A_s$
    \item $B_s' = z_s \cdot P_{net} - ch'_s \cdot (C_2 - P_s) = (r_z + ch_s \cdot r_{dean}) \cdot P_{net} - ch_s \cdot (P_s + r_{dean} \cdot P_{net} - P_s) = r_z \cdot P_{net} = B_s$
\end{itemize}
Reconstructing these four points guarantees that $ch_{s+1}' == ch_{s+1}$. This successfully closes the circular challenge loop, ensuring that, at any point, $ch'_{i} = ch_{i}$ and $ch'_{1} = ch_{1}$.
\\
\textit{Proof for $Deanonymize$:}
Before the $Deanonymize$ algorithm is executed, $k$ of $N$ deanonymization nodes have to provide their transient decryption share, $\{D_j = S_{net-j}\cdot C_1\}_{j=1}^k$, over $C1$ using their secret key shares. 
The $Deanonymize$ algorithm then computes $\lambda_j$ (cf. Algorithm 4) based on the public indices $\{\omega_j\}_{j=1}^k$ of the $k$ respondent nodes and the following values:  \\
$V = \sum_{j=1}^k \lambda_j \cdot D_j = \sum_{j=1}^k \lambda_j \cdot (S_{net-j}\cdot C_1) = \sum_{j=1}^k (\lambda_j \cdot S_{net-j}) \cdot C_1 = S_{net} \cdot C_1$\\ 
$C_1 = r_{dean} \cdot G \Rightarrow V = S_{net} \cdot (r_{dean} \cdot G) = r_{dean} \cdot (S_{net} \cdot G) = r_{dean} \cdot P_{net} $\\
$C_2 = P_s + r_{dean} \cdot P_{net} \Rightarrow C_2 - V = P_s + r_{dean} \cdot P_{net} - r_{dean} \cdot P_{net} = P_s$\\
Furthermore, \textbf{Lemma 6} proves that the output of $Deanonymize$ if a signature is valid can only be $P_s$.
This satisfies the correctness of the $Deanonymize$ algorithm. \\ 

\subsection{Proof of Lemma 2 (Signer indistinguishability)}
\label{app:proof_lemma_2}
$\{x_i,z_i\}_{i=1, i \neq s}^{n}$ the non-signer scalars are uniformly chosen randoms from $\mathbb Zq$. Their associated commitments $(L_i,R_i,A_i,B_i)$ are all computed using the $(x_i,z_i)$ scalars, so they are uniformly distributed points of $\mathbb G$.

$(L_s,R_s,A_s,B_s)$ the signer's initial commitments are computed directly using the nonces $r,r_z$ that are uniformly and randomly chosen from $\mathbb Z_q$, so they and their corresponding (equal) reconstructed commitments during verification $(L'_s,R'_s,A'_s,B'_s)$ are all uniformly distributed points of $\mathbb G$.

As for $(x_s,z_s)$, they are generated as $x_s = r - ch_s \cdot S_s \pmod q$ and $z_s = r_z + ch_s \cdot r_{dean} \pmod q$. Since the nonces $r,r_z$ and $r_{dean}$ are uniformly distributed in $\mathbb Zq$, the resulting signer responses are statistically indistinguishable from the rest of the responses.

$\Rightarrow \{x_i,z_i\}_{i=1}^{n}$ are uniformly distributed scalars of $\mathbb Z_q$, and $\{(L_i,R_i,A_i,B_i)\}_{i=1}^n$ are uniformly distributed points of $\mathbb G$. $\Rightarrow \mathcal{A}$ has no advantage from the analysis of responses and commitments.
\\
The deanonymization components $(C_1,C_2)$ are standard ElGamal encryption components. They are generated from the nonce $r_{dean}$ uniformly and randomly chosen from $\mathbb Zq$ as: $C_1 = r_{dean} \cdot G$ and $C_2 = P_s + r_{dean} \cdot P_{net}$. Since $P_{net} = S_{net} \cdot G$ and $C_2$ contains the signer's public key masked with $r_{dean} \cdot P_{net}$, for $\mathcal{A}$ to determine if $b=0$ or $b=1$ in $C_2 = P_{jb}^* + (r_{dean} \cdot P_{net})$, they need to tell $(G,S_{net} \cdot G, r_{dean} \cdot G, r_{dean} \cdot S_{net} \cdot G)$ from a uniformly random tuple $(G,S_{net} \cdot G, r_{dean} \cdot G, c \cdot G)$. $\mathcal{A}$ has a negligible advantage in this decision under the DDH assumption. Consequently, the mask $r_{dean} \cdot P_{net}$ is indistinguishable from a uniformly random point in $\mathbb G$. $\Rightarrow \mathcal{A}$ has no advantage from the analysis of $(C_1,C_2)$.

Finally, $\mathcal{A}$ was restricted from querying $O_s$ for $(P_{j_0}^*,SID^*)$ and $(P{j_1}^*,SID^*)$ so the element $I_{scope}$ in $\sigma^*$ offers no advantage to guessing the signer.

$\Rightarrow \boxed{Adv_{\mathcal{A}}^{indist} \approx 0}$

\subsection{Proof of Lemma 3 (Unforgeability - EUF-CMA)}
\label{app:proof_lemma_3}
We prove using the Forking Lemma for ring signatures \cite{forking-lemma} that if $\mathcal{A}$ has a non-negligible advantage $Adv^{EUF-CMA}_{\mathcal{A}}$, a simulator $\mathcal{B}$ can construct an algorithm utilizing $\mathcal{A}$ to solve an instance of the ECDLP with an advantage $Adv^{ECDLP}_{\mathcal{B}}$.

\underline{Simulator setup:} $\mathcal{B}$ is presented with an ECDLP challenge. A random point (public key) $P \in \mathbb G$ and the generator $G$ are provided to $\mathcal{B}$ that aims to compute $S \in \mathbb Z^*_q$ such that $P = S \cdot G$. $\mathcal{B}$ initializes the DSLRS system environment for $\mathcal{A}$:
\begin{itemize}
    \item $\mathcal{B}$ selects $v \in \{1,\cdots,K\}$ uniformly at random.
    \item $\mathcal{B}$ embeds the ECDLP challenge into the global registry $\mathbb L$ by setting the public key at the chosen index $v$ to the challenge point: $P_v = P$. The rest of the $(K-1)$ public keys are generated normally as $(S_i,P_i)$.
    \item $\mathcal{B}$ provides $\mathcal{A}$ with $PP$ which includes $\mathbb L$.
    \item $\mathcal{B}$ simulates the hash functions $H$ and $H_p$ as Random Oracles and in a consistent way. $\mathcal{B}$ also simulates a signing oracle $O_s$ and generates a valid-looking signature for every query $(m,L,P_i,SID)$.
\end{itemize}

\underline{Application of the Forking Lemma:}
$\mathcal{B}$ executes $\mathcal{A}$ twice on a specific random tape. $\mathcal{A}$ is allowed to choose the message $m^*$, the scope $SID^*$ and the ring $L^* \subset \mathbb L$. Suppose that $P_v \in L^*$. $\mathcal{A}$ executes in a time $t$ and makes $q_H$ queries to the random hash oracle $H$ and $q_{O_S}$ queries to the signing oracle $O_s$. $\mathcal{B}$ runs $\mathcal{A}$ twice:
\begin{enumerate}
    \item \textbf{First execution:} $\mathcal{A}$ outputs a valid forgery $\sigma^*$ for the signer index $s$ and $\mathcal{B}$ records all the challenges $\{ch_i\}_{i=1}^n$ outputted by the random oracle $H$. $\mathcal{B}$ also records all the internal generated randoms $\{x_i,z_i\}_{i=1}^n$.
    \item \textbf{Second execution (the fork):} $\mathcal{B}$ restarts $\mathcal{A}$ with the same random tape numbers but rewinds and forks (stops) the simulation at a uniformly random query to the random oracle $H$. Let $i^* \in \{1,\cdots,q_H$\} denote the index of the hash query to the random hash oracle $H$ that determines the signer's index challenge $ch_s$ in the forged signature. At this index, $\mathcal{B}$ provides a new randomly selected and different challenge $ch'_s \neq ch_s$. $\mathcal{A}$ outputs a second valid forgery $\sigma'^*$ using the same signer index $s$. The probability of $\mathcal{B}$ forking at $i^* = s$ is at least $\frac{1}{q_H}$.
\end{enumerate}
Both signatures $\sigma^*$ and $\sigma'^*$ are valid DSLRS signatures and share the same challenges and randoms (commitments $x_i,z_i$) up to the fork - all known to $\mathcal{B}$. Using $Verify$, $\mathcal{B}$ reconstructs at the fork $i^* = s$:
$L_s^{\sigma^*} = x_s \cdot G + ch_s \cdot P_s$ and $L_s^{\sigma'^*} = x'_s \cdot G + ch'_s \cdot P_s$. Equating these: $(x_s - x'_s) \cdot G = (ch'_s - ch_s) \cdot P_s$ and since at the fork $ch'_s \neq ch_s$ we have: $P_s = \frac{x_s - x'_s}{ch'_s - ch_s} \cdot G$. Since $P_s = S_s \cdot G$, $S_s = \frac{x_s - x'_s}{ch'_s - ch_s} \pmod q$ is extracted.
The probability that the  embedded ECDLP challenge index $v$ matches the signer index $s$ chosen by $\mathcal{A}$ for the forgery is $\frac{1}{n}$.

According to the General Forking Lemma, the \textbf{simplified} probability that $\mathcal{B}$ successfully forces $\mathcal{A}$ to produce two related forgeries and extracts the secret key $S = S_v$ is: $Adv^{ECDLP}_{\mathcal{B}} \geq \frac{1}{n} \cdot \frac{1}{q_H} \cdot (Adv^{EUF\_CMA}_{\mathcal{A}})^2$.
\\
Under the ECDLP hardness assumption $Adv^{ECDLP}_{\mathcal{B}} \approx 0$ $ \Rightarrow \frac{1}{n} \cdot \frac{1}{q_H} \cdot (Adv^{EUF\_CMA}_{\mathcal{A}})^2 \approx 0$. Since both $q_H$ and $n$ are bounded integers $\Rightarrow \boxed{Adv^{EUF\_CMA}_{\mathcal{A}} \approx 0}$ \\

\textbf{Proof of corollary 1:} $\mathcal{A}$ submits a forged signature $\sigma^*$ containing an $I_{scope} = S_s \cdot H_p(P_s||SID)$ that belongs to a signer $s$. During verification of $\sigma^*$ using $Verify$, the commitment $R'_s = x_s \cdot H_p(P_s||SID) + ch_s \cdot I_{scope}$ must match the initial commitment made during $Sign$ which is $R_s = r \cdot H_p(P_s||SID)$ (from \textbf{Lemma 1}). This requires $\mathcal{A}$ to input the response $x_s = r - ch_s \cdot S_s \pmod q$. Knowledge of $I_{scope}$ is not enough to compute a valid $x_s$, therefore $\mathcal{A}$ cannot produce a valid signature without extracting $S_s$ from $I_{scope}$, which is negligible under the ECDLP assumption.

\textbf{Proof of corollary 2:} $\mathcal{A}$ forges a signature $\sigma^*$ for a scope $SID^*$ containing a malformed key image $I'_{scope} \neq S_s \cdot B \cdot G$ with $B \cdot G = H_p(P_s||SID^*)$. Let $I'_{scope} = A \cdot G$. $A,B \in \mathbb Z_q^*$.
During $Verify$ and since $\sigma^*$ is valid, the reconstructed commitment $R'_s = x_s \cdot H_p(P_s||SID^*) + ch_s \cdot I'_{scope} = x_s \cdot B \cdot G + ch_s \cdot A \cdot G$ must match the initial $R_s = r \cdot H_p(P_s||SID^*) = r \cdot B \cdot G$. $\mathcal{A}$ must compute $x_s$ to satisfy $x_s \cdot B \cdot G + ch_s \cdot A \cdot G = r \cdot B \cdot G$.
Simultaneously, the commitment $L'_s = x_s \cdot G + ch_s \cdot P_s$ must match the initial commitment $L_s = r \cdot G$. $\mathcal{A}$ must compute $x_s$ to satisfy $x_s \cdot G + ch_s \cdot P_s = r \cdot G$. 
Extracting $x_s$ from both expressions: $x_s = r - ch_s \cdot S_s = r - ch_s \cdot \frac{A}{B}$.
$\Rightarrow A = S_s \cdot B$ which means $I'_{scope} = A \cdot G = S_s \cdot B \cdot G$. This contradicts $I'_{scope} \neq S_s \cdot B \cdot G$ so we prove that a valid DSLRS signature $\sigma$ must have a key image formed as $S_s \cdot H_p(P_s||SID)$.

\subsection{Proof of Lemma 4 (Scoped-linkability)}
\label{app:proof_lemma_4}
Case A. $P_{s1}^* = P_{s2}^*$ and $Link(\sigma_1^*,m_1^*,\sigma_2^*,m_2^*,PP) = 0$: $Verify(m_1^*,\sigma_1^*,PP) = 1$ and $Verify(m_2^*,\sigma_2^*,PP) = 1$ and $Link(\sigma_1^*,m_1^*,\sigma_2^*,m_2^*,PP) = 0$ so $I_{scope}^{\sigma_1^*} \neq I_{scope}^{\sigma_2^*}$.
Since $P_{s1} = P_{s2}$ and the scope $SID^*$ is the same for both signatures, $I_{scope}^{\sigma_1^*} \neq I_{scope}^{\sigma_2^*}$ holds true only if one of the signatures contains a malformed key image, which contradicts \textbf{Corollary 2}. $\Rightarrow \boxed{Adv_{\mathcal{A}}^{SL}[\text{Case A}] \approx 0}$. \\
Case B. $P_{s1}^* \neq P_{s2}^*$ and $Link(\sigma_1^*,m_1^*,\sigma_2^*,m_2^*,PP) = 1$: $Link(\sigma_1^*,m_1^*,\sigma_2^*,m_2^*,PP) = 1$, it follows that $I_{scope}^{\sigma_1^*} = I_{scope}^{\sigma_2^*}$. From \textbf{Corollary 2}: $S_{s1} \cdot H_p(P_{s1}^*||SID^*) = S_{s2} \cdot H_p(P_{s2}^*||SID^*)$. Since $P_{s1}^* \neq P_{s2}^*$ and $H_p$ is modeled as ROM, $H_p(P_{s1}^*||SID^*)$ and $H_p(P_{s2}^*||SID^*)$ are two distinct points in $\mathbb G$.
Finding a collision so that $S_{s1} \cdot H_p(P_{s1}^*||SID^*) = S_{s2} \cdot H_p(P_{s2}^*||SID^*)$ requires $\mathcal{A}$ to find a relationship between two independent points in $\mathbb G$, which is strictly equivalent to solving ECDLP, which is negligible. $\Rightarrow \boxed{Adv_{\mathcal{A}}^{SL}[\text{Case B}] \approx 0}$

Finally $\Rightarrow \boxed{Adv_{\mathcal{A}}^{SL} \approx 0}$

\subsection{Proof of Lemma 5 (Cross-scope unlinkability)}
\label{app:proof_lemma_5}
We prove that if $\mathcal{A}$ has a non-negligible $ADV_{\mathcal{A}}^{CSU}$, a simulator $\mathcal{B}$ can use $\mathcal{A}$ to solve the DDH problem with a non-negligible advantage.
\\
\underline{Simulator setup:} $\mathcal{B}$ is a presented with a DDH challenge tuple $(G,a \cdot G, b \cdot G, Z)$ and must decide if $Z = ab \cdot G$ or $Z = c \cdot G$ (valid tuple or random tuple). Let $P_{j_0}^* = A = a \cdot G$ and $B = b \cdot G$.
$\mathcal{B}$ initializes the DSLRS system environment for $\mathcal{A}$:
\begin{itemize}
    \item $\mathcal{B}$ selects two indices $u,v \in \{1, \cdots,K\}$
    \item $\mathcal{B}$ sets $P_u = A = a \cdot G$; $a \xleftarrow[]{} \mathbb Z_q$ is unknown to $\mathcal{B}$ ($a$ is $S_u$)
    \item $\mathcal{B}$ sets $P_v = x \cdot G$; $x \xleftarrow[]{} \mathbb Z_q$ is known to $\mathcal{B}$ ($x$ is $S_v$)
    \item $\mathcal{B}$ generates the remaining $K-2$ keys normally in $\mathbb L$
    \item $\mathcal{B}$ provides $PP$ to $\mathcal{A}$. During the learning phase, $\mathcal{A}$ chooses $P_{j_0}^*,P_{j_1}^*$ and if $(P_{j_0}^* \neq P_u) \lor (P_{j_1}^* \neq P_v)$; $\mathcal{B}$ aborts the simulation. The probability of $\mathcal{A}$ picking the right key combination is $\frac{2}{K \cdot (K-1)}$; Suppose $\mathcal{A}$ picked $P_{j_0}^* = P_u,P_{j_1}^* = P_v$. The simulation continues\footnote{$\mathcal{B}$ can reduce the number of keys in $\mathbb L$ to $n_{min}$ to maximize the chances of $\mathcal{A}$ picking the needed keys.}.
    \item $\mathcal{B}$ simulates the random oracles $H,H_p$ and programs the following responses:
\begin{itemize}
    \item $H_p(P_u||SID_0^*) = r \cdot G$; $r \xleftarrow[]{R} \mathbb Z_q$ known to $\mathcal{B}$
    \item $H_p(P_u||SID_1^*) = B$
    \item $H_p(P_v||SID_1^*) = x^{-1} \cdot Z$;
\end{itemize}
\end{itemize}
\underline{Signature generation:}
$\mathcal{B}$ generates $\sigma_1$ for $P_u$ in scope $SID_0^*$. $I_{scope}^{\sigma_1}$ is computed as $S_u \cdot H_p(P_u||SID_0^*) = a \cdot r \cdot G = r \cdot A$. Since $\mathcal{B}$ knows $r,A$ it can successfully compute this without knowing $a$. The remaining signature items $(\{L_i,R_i,A_i,B_i,x_i,z_i\}_{i=1}^n,C_1,C_2)$ are faked by back-patching the random oracle $H$ to correctly close the loop. The simulated components are uniformly distributed and indistinguishable from an honest signer output, as proven in \textbf{Lemma 2}.

$\mathcal{B}$ generates $\sigma_2$ for $P_u$ or $P_v$ in scope $SID_1^*$. $\mathcal{B}$ embeds the DDH challenge in the key image $I_{scope}^{\sigma_2} = Z$ and simulates the rest of the signature components as it did for $\sigma_1$.
$\mathcal{B}$ provides $\mathcal{A}$ with $(\sigma_1,\sigma_2)$ for it to guess $b' = 0$ same signer or $b'= 1$ different signers.
\begin{enumerate}
    \item $b'=0$ $\Rightarrow \sigma_2$ is evaluated as a valid signature from $P_u$, the key image is $I_{scope}^{\sigma_2} = S_u \cdot H_p(P_u||SID_1^*) = a \cdot B = ab \cdot G$. If $\mathcal{A}$ outputs $b'=0$, $\mathcal{B}$ outputs 1 indicating it is a valid tuple.
    \item $b'=1$ $\Rightarrow \sigma_2$ is evaluated as a valid signature from $P_v$, the key image is $I_{scope}^{\sigma_2} = S_v \cdot H_p(P_v||SID_1^*) = x \cdot x^{-1} \cdot Z = Z$. If $\mathcal{A}$ outputs $b'=1$, $\mathcal{B}$ outputs 0 indicating it is a random tuple.
\end{enumerate}
We note $Adv_{\mathcal{B}}^{DDH} \geq \frac{2}{K(K-1)} \cdot Adv_{\mathcal{A}}^{CSU}$. Since $Adv_{\mathcal{B}}^{DDH} \approx 0$
 $\Rightarrow \boxed{Adv^{CSU}_{\mathcal{A}} \approx 0}$

\subsection{Proof of Lemma 6 (Accountability)}
\label{app:proof_lemma_6}
Assume that $\mathcal{A}$ generated a valid signature $\sigma^*$ with $(C_1,C_2)$ using a false identity $P_{fake}^*$ where $C_1 = r_{dean} \cdot G$ and $C_2= P_{fake}^* + r_{dean} \cdot P_{net}$. During $Verify$, $B'_s$ is constructed as
$B'_s= z_s \cdot P_{net} - ch_s \cdot (C_2 - P_s)$
$\Rightarrow B'_s= z_s \cdot P_{net} - ch_s \cdot (P_{fake}^* + r_{dean} \cdot P_{net} - P_s)$.\\ Since $z_s = r_z + ch_s \cdot r_{dean} \pmod q$
$\Rightarrow B'_s= (r_z + ch_s \cdot r_{dean}) \cdot P_{net} - ch_s \cdot (P_{fake}^* + r_{dean} \cdot P_{net} - P_s) = r_z \cdot P_{net} + ch_s \cdot r_{dean} \cdot P_{net} - ch_s \cdot P_{fake}^* - ch_s \cdot r_{dean} \cdot P_{net} + ch_s \cdot P_s = r_z \cdot P_{net} - ch_s \cdot (P_{fake}^* - P_s)$.
Since $\sigma^*$ is valid, the reconstructed commitment $B'_s$ matches $B_s = r_z \cdot P_{net}$.\\
$\Rightarrow  r_z \cdot P_{net} - ch_s \cdot (P_{fake}^* - P_s) = r_z \cdot P_{net}$
$\Rightarrow ch_s (P_{fake}^* - P_s) = 0$. $ch_s \neq 0$ $\Rightarrow (P_{fake}^* - P_s) = 0$. $\mathcal{A}$ is strictly forced to encapsulate their $P_s$ in $C_2$ $\Rightarrow \boxed{Adv_{\mathcal{A}}^{ACC} \approx 0}$

\textbf{Proof of corollary 3:} 
Assume $\mathcal{A}$ outputs a forged signature $\sigma^*$ to frame a victim $P_v^* \in L^*$.
Let $P^*_v$ be set as $P_{fake}^*$ defined in \textbf{Lemma 6}. According to lemma 6, $\mathcal{A}$ cannot generate a valid signature where $P^*_{fake} \neq P_s$. $\mathcal{A}$ cannot frame a victim without violating the accountability proven in lemma 6.

$\Rightarrow \boxed{Adv_{\mathcal{A}}^{FRAME} \approx 0}$

\subsection{Proof of Lemma 7 (Rogue key resistance)}
\label{app:proof_lemma_7}
To register a rogue key $P_{rogue} = P_1 - P_2$, $\mathcal{A}$ is required to submit a valid Schnorr PoP $\pi_{rogue}$ on $m_{rogue} = H(P_{rogue}||SID_1)$. This requires $\mathcal{A}$ to know $S_{rogue}$ since Schnorr signatures are EUF-CMA.
By definition, $P_{rogue} = P_1 - P_2 = S_1 \cdot G - S_2 \cdot G = (S_1 - S_2) \cdot G$. Let $S_{rogue}=S_1 - S_2$, it is computationally infeasible for $\mathcal{A}$ to compute $S_1 - S_2$ under the ECDLP hardness assumption.

\end{document}